\documentclass[12pt]{article}
\newcommand{\bfm}[1]{\mbox{\boldmath{$#1$}}}
\newcommand{\beq}{\begin{eqnarray}}
\newcommand{\eeq}{\end{eqnarray}}
\newcommand{\beqs}{\begin{eqnarray*}}
\newcommand{\eeqs}{\end{eqnarray*}}

\usepackage{graphics,epsfig}
\usepackage{latexsym}
\usepackage{amsthm}
\usepackage{amssymb,amsfonts,amsmath}

\newtheorem{theorem}{Theorem}

\newtheorem{definition}{Definition}

\usepackage{xcolor}

\begin{document}

\title{Constraints on Bounded Motion and Mutual Escape for the Full 3-Body Problem}
\author{D.J. Scheeres\\Department of Aerospace Engineering Sciences\\The University of Colorado at Boulder\\scheeres@colorado.edu}
\date{\today}
\maketitle

\begin{abstract}

When gravitational aggregates are spun to fission they can undergo complex dynamical evolution, including escape and reconfiguration. Previous work has shown that a simple analysis of the full 2-body problem provides physically relevant insights for whether a fissioned system can lead to escape of the components and the creation of asteroid pairs. In this paper we extend the analysis to the full 3-body problem, utilizing recent advances in the understanding of fission mechanics of these systems. Specifically, we find that the full 3-body problem can eject a body with as much as 0.31 of the total system mass, significantly larger than the 0.17 mass limit previously calculated for the full 2-body problem. 
This paper derives rigorous limits on a fissioned 3-body system with regards to whether fissioned system components can physically escape from each other and what other stable relative equilibria they could settle in. We explore this question with a narrow focus on the Spherical Full Three Body Problem studied in detail earlier. 

\end{abstract}

\section{Introduction}

When a gravitational aggregate spins fast enough to fission or shed material from its surface, a fundamental question is whether components in the newly formed system in relative orbit about each other can undergo a mutual escape. This is a known process in nature, with the discovery of asteroid pairs \cite{vokrouhlicky2008pairs} and their linkage to formation from a rotational fission event \cite{pravec_fission}. The current paper extends previous analysis of this process, which was focused on the fission of contact binary bodies \cite{scheeres_F2BP, scheeres_icarus_fission, scheeres_F2BP_planar}, and considers \textcolor{black}{the next step in generality, specifically } this paper considers the fission of a full 3-body system of unequal masses. In so doing, we discover situations under which the fission limits on rigid 2-body systems can be violated. 

The specific problem studied in this paper belongs to the class of full body problems, which studies the dynamics of gravitational attracting rigid bodies that can rest on each other. Several different aspects of this problem have been studied, including the general problem of 2 bodies \cite{maciejewski, wang, cendra_marsden, scheeres_F2BP, scheeres_icarus_fission, scheeresPSS, scheeres_F2BP_planar}, the $N$-body problem of equal-sized spheres \cite{scheeres_minE, IAU_namur, IAU_hawaii, F4BP_chapter}, and most recently a general analysis of the spherical 3-body problem with arbitrary values of mass and size \cite{F3BP_scheeres}. While of interest mathematically, the study of this problem has also proven to be relevant to understanding the physical evolution of rubble pile asteroids, primitive solar system bodies that consist of self-gravitating boulders and grains. These bodies are subject to unique physical evolutionary forces that can cause their spin rates, and hence total angular momentum, to slowly change over time as an adiabatic process (reviewed in \cite{scheeres_minE}). This creates a situation where the total angular momentum and energy of the self-gravitating system can increase in time, eventually reaching a transition in a stable resting energy state and triggering a profound change in the system configuration. 

A key outcome once the system undergoes a change in system configuration is whether one of the bodies can be ejected from the system. This corresponds to the Hill Stability of the system, delineating whether or not the system can have one component that departs to infinity relative to the others. For the full 2-body problem, a general analysis shows that if the fissioned binary system has one component which has a mass less than $\sim$0.17 of the total system mass, then the system will have a positive energy after fission and can result in the two bodies mutually escaping, with the orbital energy coming from transferred rotational kinetic energy via gravitational torques (for systematic studies of this problem in terms of energy and angular momentum transfer see \cite{jacobson_icarus, boldrin_mnras}). This result has been shown to be consistent with and apparently occur in solar system bodies, with the discovery of asteroid pairs (asteroids that in the past appear to have emanated from the same locale with low relative speeds) and the correlation of their spin properties and composition with the full 2-body problem fission theory \cite{pravec_fission, polishook2014observations, polishook2014rotationally}. 

In this paper we consider conditions under which a system of three bodies in a stable resting configuration brought to fission spin rates can partially or fully escape. Also related to this is the identification of when, post fission, other possible stable states exist. This study builds on a recent paper \cite{F3BP_scheeres} identifying all stable minimum energy states in the full 3-body problem and finding conditions under which they fission or become unstable. 
Our rigorous analysis shows that some stable full 3-body relative equilibrium can eject a body of mass up to 0.31 of the total system mass when it undergoes  fission \textcolor{black}{with the mass being one of the single components}. This expanded mass fraction could explain the few, but significant, more recently found asteroid pairs that do not conform to the limits found from the 2-body problem (see Fig.\ 6 in \cite{AIV_interiors})

The outline of this paper is as follows. In Section 2 we define the spherical full 3-body problem and set up important definitions and concepts used throughout the paper. In Section 3 we define and prove a few fundamental results that enable our direct computations. In Section 4 the main results are presented. Section 5 discusses the implications and applications of these results. 

\section{Background}

\subsection{The Spherical, Finite Density 3-Body Problem}

In \cite{F3BP_scheeres} the Spherical, Finite Density 3-body problem (SF3BP) is introduced and analyzed in detail for its relative equilibria and their stability. This problem is a variation of the 3-body problem in that the massive bodies have a finite density, meaning that they have a finite physical extent and thus cannot come arbitrarily close to each other. Due to this, it is possible to have resting configurations as possible relative equilibria. The spherical restriction for this problem means that the bodies are assumed to be spheres, which simplifies the computation of the mutual gravitational potential and surface contact conditions. As another note, the assumption is made that the bodies in contact will experience frictional forces, and thus a system in a relative equilibrium will have no motion relative to each other at their contact point, meaning that at a relative equilibria member bodies will rotate in concert. This problem has been specifically posed and studied in a series of papers that have focused on non-spherical bodies and equal sized bodies. The current paper continues the study to issues of Hill stability and other possible end states for systems that undergo a transition in their stability due to a increase in the system angular momentum. 

The use of spherical bodies is an idealization, as for transfer of angular momentum and energy from rotational to translational motion it is crucial for the bodies to be non-spherical. However, previous research for the full 2-body problem \cite{scheeres_F2BP_planar} has definitely shown that the use of spherical bodies correctly captures the energetics and angular momentum of these dynamically evolving systems. This motives our current study of this simplified system. 

Consider three bodies, $\mathcal{ B}_i, i = 1,2,3$, each of which is a sphere of radius $R_i$ and, for convenience, assumed to have a common density $\rho$, giving each of them a mass of $M_i = 4\pi/3 \ \rho R_i^3$. 
The positions of these bodies can be denoted in $\mathbb{R}^3$ by Cartesian position vectors $\bfm{D}_i$. The relative positions of these bodies are denoted as $\bfm{D}_{ij} = \bfm{D}_j - \bfm{D}_i$ and have the fundamental rigid body constraint $| \bfm{D}_{ij} | \ge (R_i + R_j)$ for $i\ne j$. This lower bound, due to the bodies having finite density, is what enables resting equilibria to occur. Each of the spheres can carry angular momentum in their spin rate, although due to their symmetry the specific orientation of these spheres are arbitrary in any frame. Thus, the internal relative configuration space of the system, $\bfm{Q}$, is completely specified by only three quantities
\beq
	\bfm{Q} & = & \left\{ D_{12}, D_{23}, D_{31} \ | \  \right. \nonumber \\
		& & \left. D_{ij} \ge (R_i + R_j)  \ \& \ |D_{ij} - D_{jk}| \le D_{ki} \le |D_{ij} + D_{jk}| \right\}
\eeq

The existence and stability of relative equilibria can be analyzed and determined through finding the singular points and manifold curvature of the amended potential defined in \cite{scheeres_minE} 
\beq
	\mathcal{E}(\bfm{Q}) & = & \frac{H^2}{2 I_H(\bfm{Q})} + \mathcal{U}(\bfm{Q})
\eeq
where $H$ is the total angular momentum magnitude of the system, $I_H$ is the total system moment of inertia about the system angular momentum vector, and $\mathcal{U}$ is the gravitational potential energy of the system. The above amended potential is a lower bound on the total system energy, 
\beq
	\mathcal{E}(\bfm{Q}) & \le & E
\eeq
where $E = T_o + T_r + \mathcal{U}$, where $T_o$ is the translational kinetic energy, $T_r$ is the rotational kinetic energy and $\mathcal{U}$ is the gravitational potential energy. 
This constraint, including its sharpness, has been proven previously \cite{scheeres_minE, F3BP_scheeres} and is used extensively in the following. 

Based on the analysis in \cite{F3BP_scheeres} we only consider configurations where the three bodies and their relative velocities all lie in a common plane. Then the moment of inertia and gravitational potential energy are defined as
\beq
	I_{H} & = & \frac{1}{M_1 + M_2 + M_3} \left[ M_1 M_2 D_{12}^2 + M_2 M_3 D_{23}^2 + M_3 M_1 D_{31}^2\right] + I_S \\
	I_S & = & \frac{2}{5} M_1 R_1^2 + \frac{2}{5} M_2 R_2^2 + \frac{2}{5} M_3 R_3^2 \\
	\mathcal{U} & = & - \mathcal{ G} \left[ \frac{M_1 M_2}{D_{12}} + \frac{M_2 M_3}{D_{23}} + \frac{M_3 M_1}{D_{31}} \right] 
\eeq
where $\mathcal{G}$ is the gravitational constant. 

To simplify the discussion, normalize the system with a length and a mass scale. The length scale used is $R_T = R_1 + R_2 + R_3$, while the mass scale is $M_T = M_1 + M_2 + M_3$. Denote $m_i = M_i / M_T$, $r_i = R_i / R_T$, and $d_{ij} = D_{ij} / R_T$. In normalized coordinates the fundamental quantities take on the values
\beq
	\mathcal{U} & = & - \left[ \frac{m_1 m_2}{d_{12}} + \frac{m_2 m_3}{d_{23}} + \frac{m_3 m_1}{d_{31}} \right] \\
	I_{H} & = & m_1 m_2 d_{12}^2 + m_2 m_3 d_{23}^2 + m_3 m_1 d_{31}^2 + I_S \label{eq:IH} \\
	I_S & = & \frac{2}{5} m_1 r_1^2 + \frac{2}{5} m_2 r_2^2 + \frac{2}{5} m_3 r_3^2
\eeq
with the angular momentum being normalized by the dividing factor $\sqrt{\mathcal{ G} M_T^3 R_T}$ and the energy normalized by the dividing factor $\mathcal{ G} M_T^2 / R_T $. For $H$, $\mathcal{E}$ and $E$ the same notational designation is kept for the normalized values.  

The normalizations provide two identities:
\beq
	r_1 + r_2 + r_3 & = & 1 \\
	m_1 + m_2 + m_3 & = & 1
\eeq
There are also fundamental relationship between the $r_i$ and the $m_i$, assuming constant density. 
\beq
	m_i & = & \frac{r_i^3}{r_1^3 + r_2^3 + r_3^3} \label{eq:massnorm} \\
	r_i & = & \frac{m_i^{1/3}}{m_1^{1/3} + m_2^{1/3} + m_3^{1/3}}
\eeq
With these identities the normalized configuration space can be stated as
\beq
	\bfm{q} & = & \left\{ d_{12}, d_{23}, d_{31} \ | \ d_{ij} \ge 1-r_k  \ \& \ |d_{ij} - d_{jk}| \le d_{ki} \le |d_{ij} + d_{jk}| \right\}
\eeq

Due to the symmetry of the problem the study can be restricted to 
\beq
	0 \le m_3 \le m_2 \le m_1 \le 1 \\
	0 \le r_3 \le r_2 \le r_1 \le 1 
\eeq
This region is shaded in Fig.\ \ref{fig:triangle}. There are 5 other equivalent triangles defined by reordering the different inequalities given above. The approach taken will be to exhaustively study all possible relative equilibria in the denoted region, the results of which can then be easily applied to all other regions. 

\begin{figure}[h!]
\centering
\includegraphics[scale=0.35]{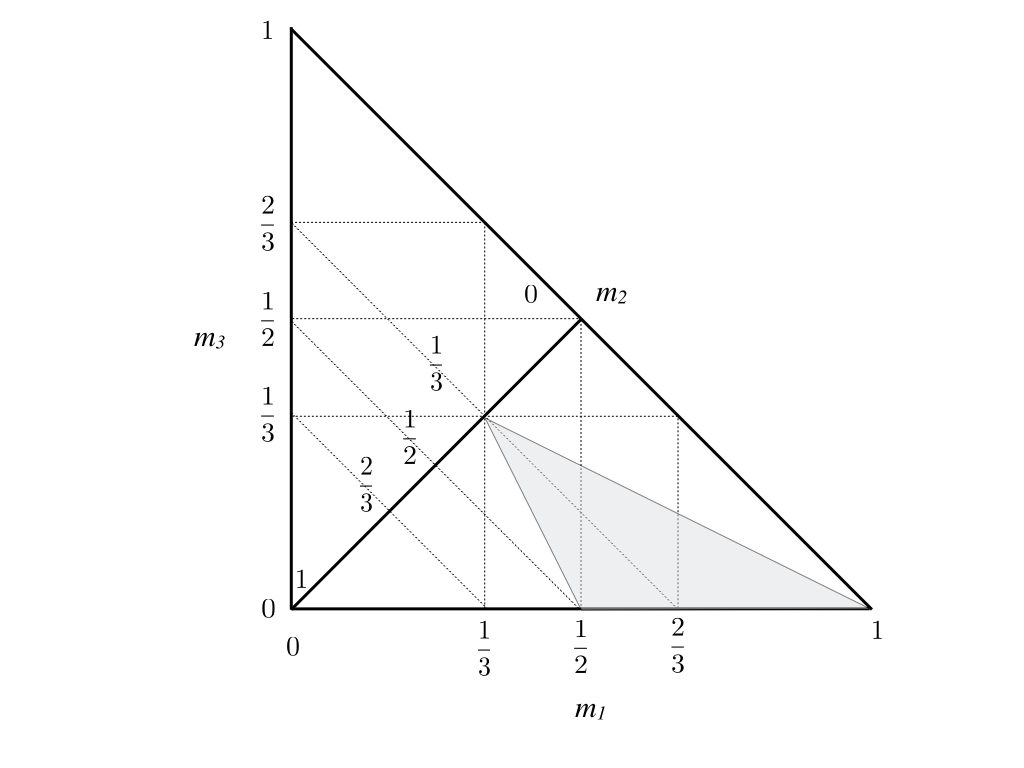}
\caption{Triangle defined for the bodies with the region of study shaded. Mass $m_1$ is measured along the horizontal axis, $m_3$ along the vertical, and $m_2$ along the diagonal, increasing from zero at the hypotenuse to unity at the origin.}
\label{fig:triangle}
\end{figure}

With this convention, there are additional constraints for the masses and radii.
\beq
	\frac{1}{3} \le (r_1, m_1) \le 1 \\
	0 \le (r_3, m_3) \le \frac{1}{3}  \\
	0 \le (r_3, m_3) \le (r_2, m_2) \le \frac{1}{2}
\eeq

\subsection{Stable Relative Equilibria and Fission}

The full 3-body problem was analyzed for all relative equilibria in \cite{F3BP_scheeres}, including the case of resting equilibria, along with the specific points at which stable relative equilibria can bifurcate into or out of existence. In the following a few important points are recounted.
Beyond the existence, stability and relative bifurcation pathways for these stable relative equilibria, whether or not they exist at the fission of separate types of configuration was not explored. This is discussed as a major component of this paper. 

\subsubsection{Lagrange Resting Equilibria}

In general it was found that at low angular momentum values only the Lagrange Resting (LR) configuration exists, shown in Fig.\ \ref{fig:resting}. 
The relative distances between the bodies satisfy the equality $d_{ij} = 1 - r_k$. 
As the angular momentum is increased, at the specific angular momentum and energy
\beq
	H^2_{LR} & = & \frac{I_{H_{LR}}^2}{(1-r_3)^3} \label{eq:LRH2} \\
	\mathcal{E}_{LR} & = & \frac{I_{H_{LR}}}{2(1-r_3)^3} + \mathcal{U}_{LR} \label{eq:LRE}
\eeq
the LR configuration ceases to exist and the configuration will fly apart and start a period of chaotic dynamical evolution. 
The associated spin rate of the system is
\beq
	\Omega^2_{LR} & = & \frac{1}{(1-r_3)^3} \label{eq:LRFSR}
\eeq

The initial disruption of the LR configuration occurs with the larger two grains separating, rotating about the smaller grain. However, once this initial separation occurs, the centrifugal forces between the remaining bodies in contact are exceeded, causing the entire system to fall apart and undergo a period of reimpact and dynamical evolution. 

\begin{figure}[h!]
\centering
\includegraphics[scale=0.35]{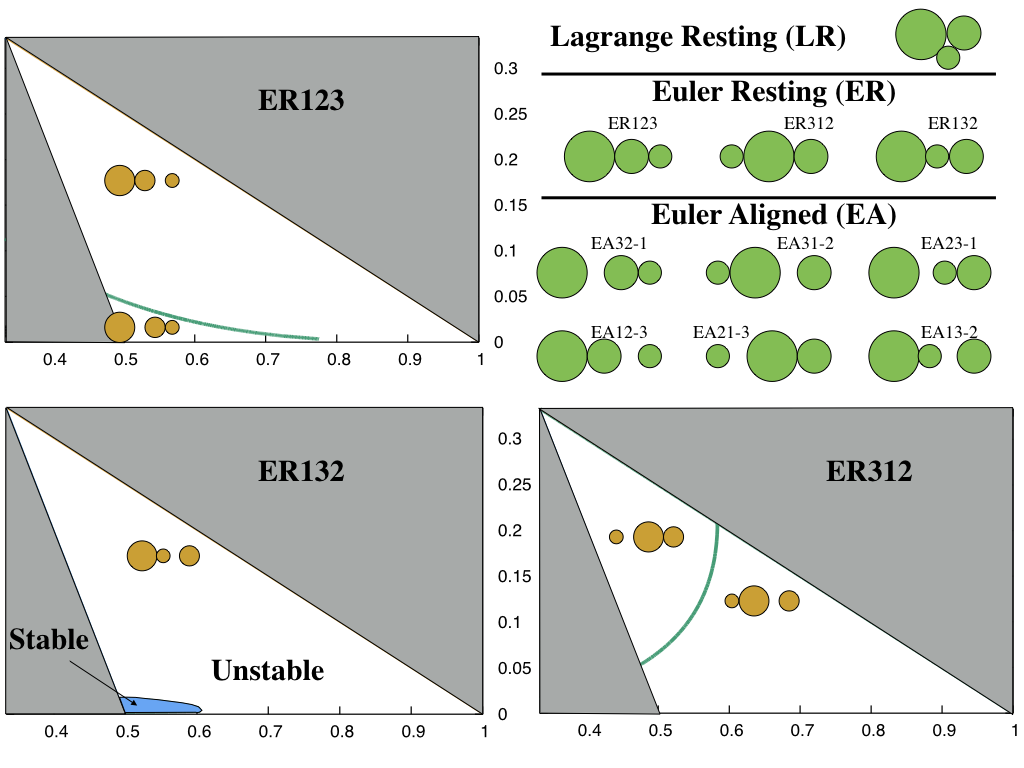}
\caption{Stable resting equilibria and their fission patterns. Note that the ER132 configuration fissions into a stable relative equilibrium in the region indicated,  the only instance of such an outcome across all of the resting equilibria. All other fissions immediately enter an unstable environment.}
\label{fig:resting}
\end{figure}

\subsubsection{Euler Resting Equilibria}

At specific non-zero levels of angular momentum, it is possible for the grains to begin to stably rest on each other in a line, what we call the Euler Resting (ER) configurations (see Fig.\ \ref{fig:resting}). If the bodies lie in sequence $i$, $j$ and $k$, with body $j$ between bodies $i$ and $k$, we call this an ERijk configuration. The relative distances satisfy $d_{ij} = 1 - r_k$, $d_{jk} = 1 - r_i$ and $d_{ki} = 1 + r_j$. 

These configurations (denoted with the subscript $ijk$) will be stable when the angular momentum satisfies  
\beq
	H^2 & \ge & \frac{I_{H_{ijk}}^2}{(1+r_j)^3} \label{eq:ERX}
\eeq
and will be stable up to 
\beq
	H^2_{ijk} & = & \frac{I_{Hijk}^2 }{(1+r_j)^3} \min\left[ \frac{1 + \frac{m_j}{m_i}\left(\frac{1+r_j}{1-r_i}\right)^2}{1 + \frac{m_j}{m_i}\frac{1-r_i}{1+r_j} },   \frac{1 + \frac{m_j}{m_k}\left(\frac{1+r_j}{1-r_k}\right)^2}{1 + \frac{m_j}{m_k}\frac{1-r_k}{1+r_j} } \right]  \label{eq:ERH2} \\
	\mathcal{E}_{ijk} & = &  \frac{H^2_{ijk}}{2 I_{Hijk}} + \mathcal{U}_{ijk} \label{eq:ERE} 
\eeq
at which point they separate and fission. 
The spin rate at which fission occurs is then
\beq
	\Omega^2_{ijk} & = & \frac{1}{(1+r_j)^3} \min\left[ \frac{1 + \frac{m_j}{m_i}\left(\frac{1+r_j}{1-r_i}\right)^2}{1 + \frac{m_j}{m_i}\frac{1-r_i}{1+r_j} },   \frac{1 + \frac{m_j}{m_k}\left(\frac{1+r_j}{1-r_k}\right)^2}{1 + \frac{m_j}{m_k}\frac{1-r_k}{1+r_j} } \right] \label{eq:ERFSR}
\eeq

Note that for 3 bodies there are 3 unique orderings of the bodies, shown in Fig.\ \ref{fig:resting}, ER123, ER132 and ER312. For an ERijk configuration there are two possible modes for fission, with loss of contact occurring between bodies $i,j$ or $j,k$. The diagrams in Fig.\ \ref{fig:resting} indicate the mass ratio values where the different separations occur. These figures are reprised from \cite{F3BP_scheeres}, but that paper showed the limits in terms of the radii of the bodies, whereas this figure shows them in terms of the mass of the bodies. For configuration ER123, in the top left, most of the parameter space has the smaller body separating from the larger bodies, except for the case of nearly equal-sized larger bodies and small third body, where the larger body separates. For ER312 in the lower right we note that for more equal-sized bodies the smaller will separate while for a smaller third body the intermediate body will separate. Finally, for ER132 we note that the configuration always fissions by the intermediate mass separating. 

The stability of the fissioned systems were also investigated. For ER123 and ER312 the post-fission state is always in a non-equilibrium state, meaning that the system will immediately undergo chaotic evolution. For the ER132 configuration there is a small area of parameter space where the system will fission directly into a stable orbital relative equilibrium, shown in the lower left of Fig.\ \ref{fig:resting}. 

\subsubsection{Euler Aligned Equilibria}

Also shown in Fig.\ \ref{fig:resting} are the Euler Aligned configurations. These appear in a more complex bifurcation sequence and, except for the small stable region for the ER132 fission mentioned above, always have an unstable component close to and terminating at the ER fission points, and a stable component that lies at a farther distance. As the angular momentum becomes larger the 6 distinct EA configurations are the only stable relative equilibria in the full 3-body problem, and exist as angular momentum becomes arbitrarily large. 

\section{Supporting Results}

\subsection{Classification of Final Motions}

Chazy has identified several classes of final motions in the 3-body problem, nicely summarized in \cite{arnold1988mathematical}, that we can directly borrow for our own classifications in the SF3BP. Our classification combines the concepts of hyperbolic and parabolic solutions, as for our case they yield indistinguishable outcomes. We do not consider the class of oscillatory motions as this is estimated to have measure zero \cite{arnold1988mathematical}. We note that all of the Chazy classes of motion do not necessarily live within the SF3BP as we do not include the point mass case. In addition, with the addition of rotation of the bodies, a relevant question not addressed here is whether there are additional interesting sub-classes of motion. 

\begin{definition}{\bf Motion Classifications}
Define three broad classes of motion, combined together from Chazy's general classes of motion:
\begin{enumerate}
\item
Bounded Motion ($\mathbf{B}$): $\sup_{t\ge t_0} |d_{ij}| < \infty \ \forall i,j$
\item
Hyper/Parabolic--Elliptic Motion ($\mathbf{HE}_k$): $|d_{jk}|, |d_{ki}| \rightarrow \infty$, $|\dot{d}_{jk}|, |\dot{d}_{ki}| \rightarrow c_k \ge 0$, $\sup_{t\ge t_0} |d_{ij}| < \infty$
\item
Hyper/Parabolic Motion ($\mathbf{H}$): $|d_{ij}| \rightarrow \infty$, $|\dot{d}_{ij}| \rightarrow c_{ij} \ge 0  \ \forall i,j$
\end{enumerate}
\end{definition}

Given these classifications, it is possible to identify either necessary or sufficient conditions for which these outcomes can occur. 
\begin{theorem}
\label{thm:1}
Necessary or Sufficient conditions for a Spherical Full 3-Body system to lie in each of these classes are: 
\begin{enumerate}
\item
A system is Bounded ($\mathbf{B}$) if $E < \min_{ij}\mathcal{U}_{ij}$ 
\item
A system can be $\mathbf{HE}_k$ for any k only if $E \ge \min_{ij}\mathcal{U}_{ij} \ \forall i,j$ 
\item
A system can be $\mathbf{HE}_k$ for a given k only if $E \ge \min_{ij\ne k}\mathcal{U}_{ij} $
\item
A system can be $\mathbf{H}$ only if $E \ge 0$ 
\item
A system is $\mathbf{HE}$ or $\mathbf{H}$ if $E - T_r^M > 0$, where $T_r^M = \sup_{t\ge t_0} T_r$ is an upper bound on the total rotational kinetic energy of the system. 
\end{enumerate}

\end{theorem}

\begin{proof}
The mechanics of the proof are relatively straightforward, and rely on the inequalities $\mathcal{E} \le E$ and $d_{ij} \ge r_i + r_j  = 1 - r_k > 0$, and on previous proofs. \\

We first establish an inequality on the mutual gravitational potential between two grains, $i$ and $j$: $\mathcal{U}_{ij} = - \frac{m_i m_j}{d_{ij}} \ge - \frac{m_i m_j}{1-r_k}$, which reduces to $d_{ij} \ge r_i + r_j$. \\

Next we show that whenever a single body escapes, say body $k$, then $\mathcal{E} \rightarrow \mathcal{U}_{ij}$. 
First note that the moment of inertial $I_H$ is unbounded whenever any body undergoes escape. Say that body $k$ escapes, leading to $d_{jk}, d_{ki} \rightarrow \infty$ and $\mathcal{U}_{jk}, \mathcal{U}_{ki} \rightarrow 0$. Then we have, by definition of $I_H$, that $I_H \ge m_k (m_i + m_j) \min( d_{jk}, d_{ki} )$. However, as both of these distances approach $\infty$ we see that $I_H$ must also approach infinity. Thus for a finite value of $H^2$, we have that $\frac{H^2}{2 I_H} \rightarrow 0$ and $\mathcal{E} \rightarrow \mathcal{U}_{ij}$. \\

Now we prove the different items in order. 
\begin{description}
\item[1, 2 and 3] 
Assume a system with $E < \min_{ij}\mathcal{U}_{ij}$ has a body $k$ that escapes. Since $k$ escapes we have $\mathcal{E} \rightarrow \mathcal{U}_{ij}$. However, as $\mathcal{E} \le E < \min_{ij}\mathcal{U}_{ij}$, this is a contradiction, meaning that body $k$ cannot escape if $E < \min_{ij}\mathcal{U}_{ij}$, and establishing \#1. Conversely, if instead body $k$ escapes then $ \min_{ij\ne k}\mathcal{U}_{ij} \le \mathcal{E} \le E$, establishing \#3. More generally, if any body leaves then $\min_{ij}\mathcal{U}_{ij} \le \mathcal{E} \le E$, establishing \#2. 
\item[4]
Assume a system has $\mathbf{H}$ motion, meaning that all of its components escape relative to each other. Then $\mathcal{E} \rightarrow 0 \le E$, establishing \#4. 
\item[5]
To prove this result, we appeal to the classical proof of escape in the $N$-body point-mass problem, as discussed in \cite{pollard} and applied to the full body problem in \cite{scheeres_F2BP}. From \cite{scheeres_F2BP} we note that the polar moment of inertia can be shown to equal 
\beq
	I_p & = & \frac{3}{2} I_S + \sum_{i < j} \frac{m_i m_j}{M_T} d_{ij}^2 
\eeq
As our bodies are spherical and have point mass potentials, the mutual potential is a homogenous function degree of -1. Thus, we can apply the classical result for the second time derivative of $I_p$, yielding
\beq
	\ddot{I}_p & = & 2 \sum_{i < j} \frac{m_i m_j}{M_T} v_{ij}^2 + 2 \mathcal{U} 
\eeq
For a full body problem $T_o = \frac{1}{2} \sum_{i < j} \frac{m_i m_j}{M_T} v_{ij}^2 = T - T_r \ge 0$, leading to
\beq
	\ddot{I}_p & = & 2 (T - T_r)  + 2 ( E - T_r ) 
\eeq
From the theorem statement, we note that $E - T_r \ge E - T_r^M > 0$. Further, by definition we have $T-T_r \ge 0$. Thus, the quantity $\ddot{I}_p > 0$, where $I_p > 0$ by definition. We note the need to make the assumption that $E$ is conserved, implying that there are no impacts or only conservative interactions. It is then classical to show that $I_p \rightarrow \infty$ \cite{pollard}, and thus that at least one body will escape. We note that without the conservative assumption, if two of the bodies strike each other the total energy can be decreased, potentially to the point where the condition is violated. 
\end{description}
\end{proof}

\section{Final States of Fissioned Systems in the Spherical, Full 3-body Problem}

Now consider the system energy of different termination fissions for the spherical full 3-body problem in order to constrain the possible post-fission outcomes. We will apply the conditions from Theorem \ref{thm:1} in order to map out parameter values that lead to different possible outcomes. 

\subsection{Lagrange Resting Fission}

There are a few questions of interest that can be answered. Following a fission event, and at that given level of angular momentum, when is the fissioned system bound, when could it undergo escape, and what other stable states exist that the system could settle into. These are addressed in the following.

The fission condition for the Lagrange Resting (LR) configuration are given in Eqns.\ \ref{eq:LRH2} and \ref{eq:LRE}.  Figure \ref{fig:LRF} compares the fission energy of the LR configuration with the sufficiency condition for bounded motion and the necessary conditions for the various escape conditions. The plot is generated by differencing the fission energy with the energy levels, and plotting the level surfaces where they equal each other. The computations are carried out using GNUPLOT 5.0. 
Figure \ref{fig:LRF} shows these comparisons for a fissioned LR system across the range of masses $m_3 \le m_2 \le m_1$. We note a pattern which is qualitatively repeated for the fission of the ER configurations as well. 
First, for systems with nearly equal mass we find that the fissioned system is bounded, meaning no body can escape. If the smaller body is less than 0.25 of the total system mass then it can be ejected, and bounded motion is no longer guaranteed. The necessary condition for ejection of body 2 hold when the mass of that body is small enough, represented by the long line that runs along the right, upper edge of the triangle. Only when both body 2 and 3 are smaller than $\sim 0.1$ each can body 1 be ejected instead, with this condition very close to the condition for all three bodies being able to mutually escape. We note that as the mass of body 3 goes to zero the necessary conditions for ejection all converge on a value of $\sim0.83$ along the body 1 axis, which is the classical limit for the full 2-body problem. 

\begin{figure}[htb]
\centering
\includegraphics[scale=0.35]{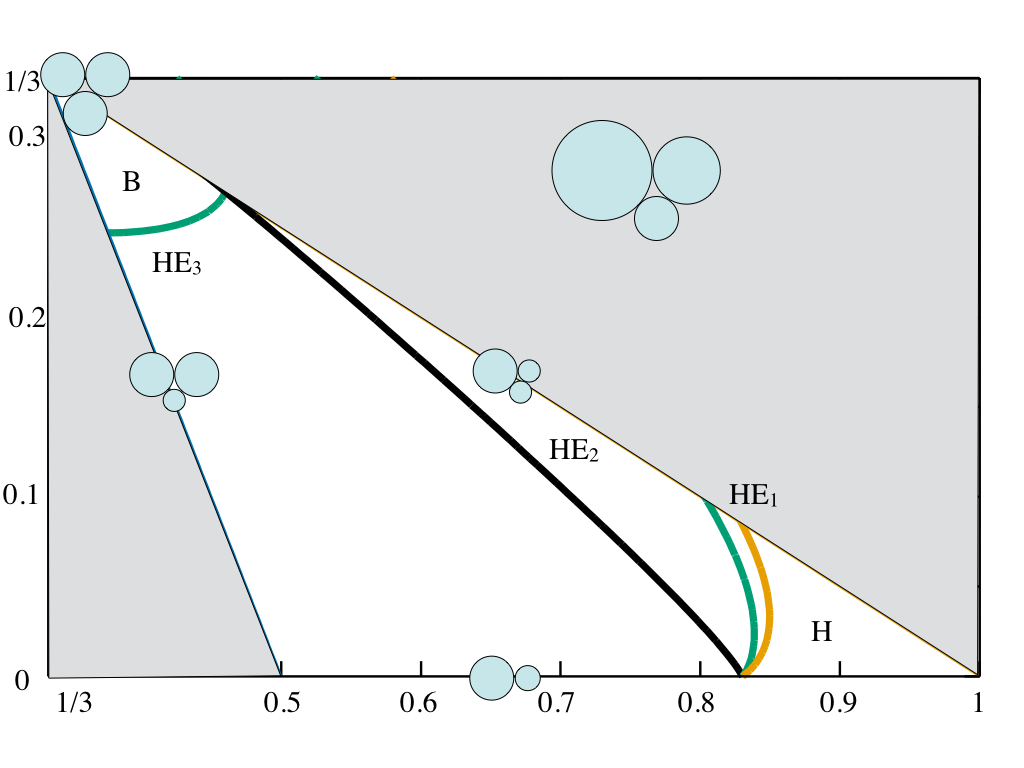}
\caption{Diagram showing mass fractions for different Hill Stability conditions to be satisfied when the Lagrange Resting configuration undergoes fission. 
\textcolor{black}{The vertical axis represents the mass fraction of the smallest component while the horizontal axis is the largest component. The different sized collections indicate the qualitative geometry of the different limiting lines on the diagram. The regions are labeled as a function of what sort of Hill stability ensures following rotational fission, as described in the text. B represents bounded motion, while HE1, HE2 and HE3 represent regions where one of the bodies can be ejected, and H represents the region of positive energy when complete escape of all bodies relative to each other can occur.  }}
\label{fig:LRF}
\end{figure}

Figure \ref{fig:LRmER} addresses the next question about whether, at the angular momentum at which the LR configuration fissions, the different ER configurations have undergone fission as of yet, or if they have stablized. If they have not fissioned yet (to the left of colored lines in the diagram), and if they have undergone stabilization (there is only one small region for the ER132 configuration where this has not occurred yet), then that ER configuration constitutes a lower-energy state that the system may be able to settle into. We note the trend that for smaller mass values for bodies 2 and 3 the LR configuration does not have a lower ER configuration to settle in. There is no strong correlation between the necessary conditions for escape and the existence of these lower energy states, indicating that there are relatively large regions of parameter space where either escape or settling of the configuration is possible. For the small region where the ER132 configuration has not stabilized yet when the LR configuration fissions, \textcolor{black}{indicated in the lower left of Fig.\ \ref{fig:LRmER},} we note that this area lies within the region where the ER132 configuration fissions directly into a stable EA13-2 relative equilibrium\textcolor{black}{, as shown in Fig.\ \ref{fig:resting} in the bottom left panel}. 

\begin{figure}[htb]
\centering
\includegraphics[scale=0.35]{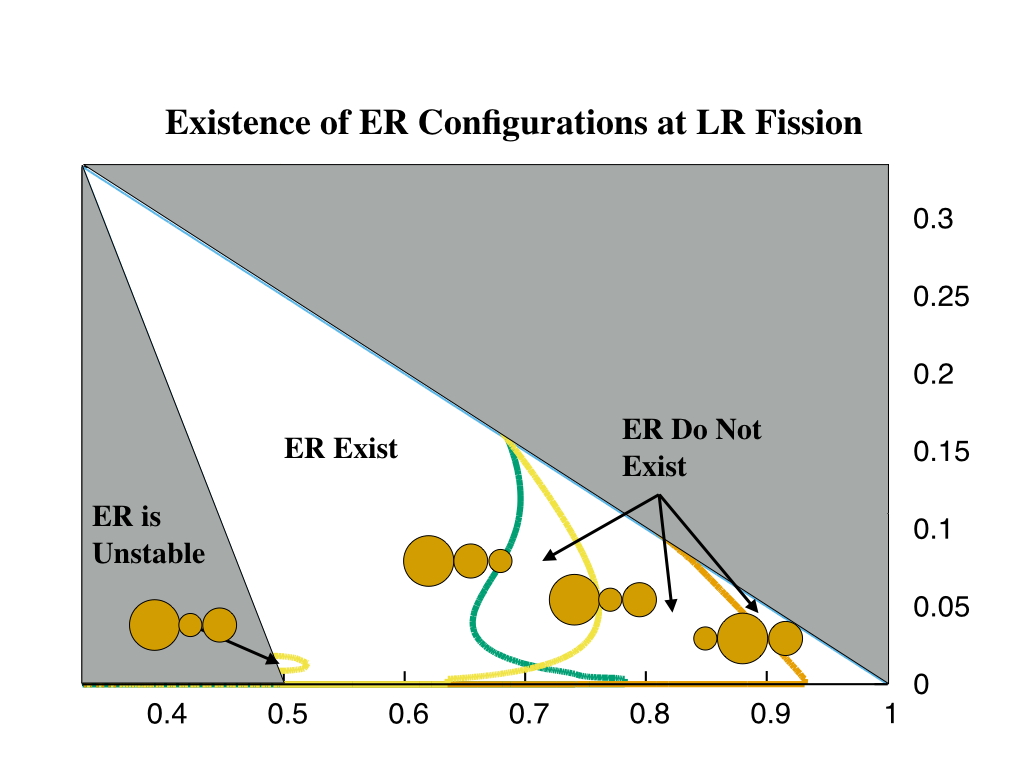}
\caption{Diagram showing whether or not a given ER configuration exists when the LR configuration fissions. }
\label{fig:LRmER}
\end{figure}

It is more difficult to ascertain whether a stable EA relative equilibrium exists across all of the LR fission points. The difficulty lies in that to determine the angular momentum at which a given EA configuration bifurcates into existence requires the iterative solution of a non-linear equation. Thus, at every point of the diagram an iterative computation would have to be made, which is possible does but does not lend itself to the computation of our diagrams. However, we know that the stable EA configurations exist at values of angular momentum lower than at which the ER fission occurs. Thus, prior to the \textcolor{black}{rotational fission transitions indicated by the} colored lines in Fig.\ \ref{fig:LRmER} the EA configuration already exists and is available as a stable final state of the system given sufficient energy dissipation. This question is addressed in a different manner later in this section. 

\subsection{Euler Resting Fission}

The fission condition for Euler Resting (ER) configurations are a bit more complex and given by Eqns.\ \ref{eq:ERH2} and \ref{eq:ERE}. The questions are more limited for these fissions, as by definition there will always be an EAij-k or EAkj-i orbital configuration that exists at the angular momentum value at which the ER configuration fissions. In Figs.\ \ref{fig:ERH123}, \ref{fig:ERH132} and \ref{fig:ERH312} we show when conditions for bounded motion and escape are satisfied for the fission of resting configurations ER123, ER132 and ER312, respectively. As noted above, these diagrams are all qualitatively similar to each other and to the LR diagram, although the specific limits shift around. Of special interest in all of these diagrams is the mass fraction of the smallest body when it is first susceptible to ejection from the system. As discussed earlier, for the full 2-body problem this ratio equals $\sim0.17$. From analysis of the Hill Stability diagrams we find that the ER132 configuration has the highest possible ejected mass fraction when it fissions, at a value of 0.311. This significantly extends the mass fraction of bodies that could fission and form asteroid pairs due to post-fission internal interactions alone. 

\begin{figure}[htb]
\centering
\includegraphics[scale=0.25]{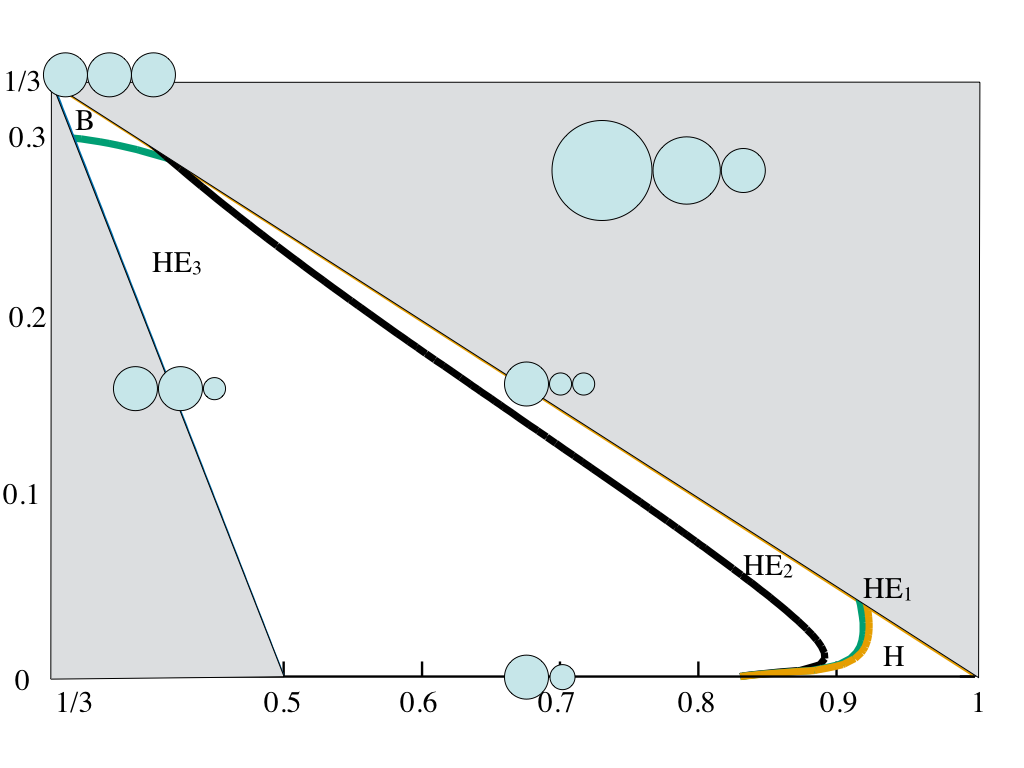}
\caption{Diagram showing mass fractions for different Hill Stability conditions to be satisfied for the ER123 configuration fission condition. The highest mass fraction that can be ejected is seen to be 0.303.}
\label{fig:ERH123}
\end{figure}

\begin{figure}[htb]
\centering
\includegraphics[scale=0.25]{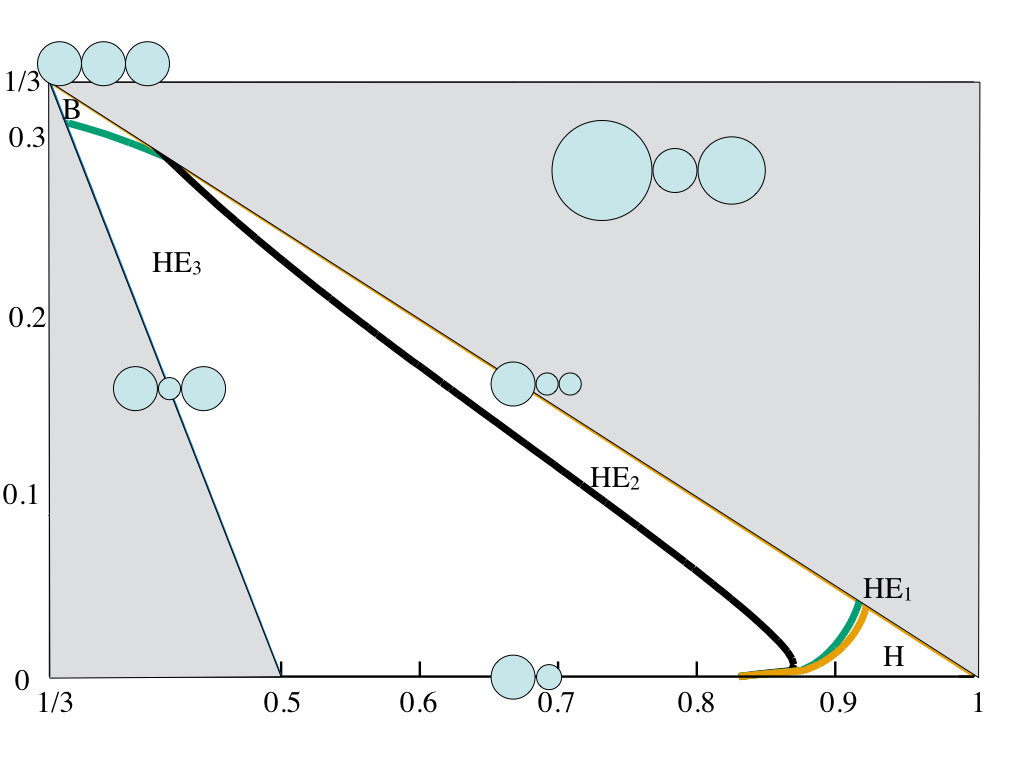}
\caption{Diagram showing mass fractions for different Hill Stability conditions to be satisfied for the ER132 configuration fission condition. The highest mass fraction that can be ejected is seen to be 0.311.}
\label{fig:ERH132}
\end{figure}

\begin{figure}[htb]
\centering
\includegraphics[scale=0.25]{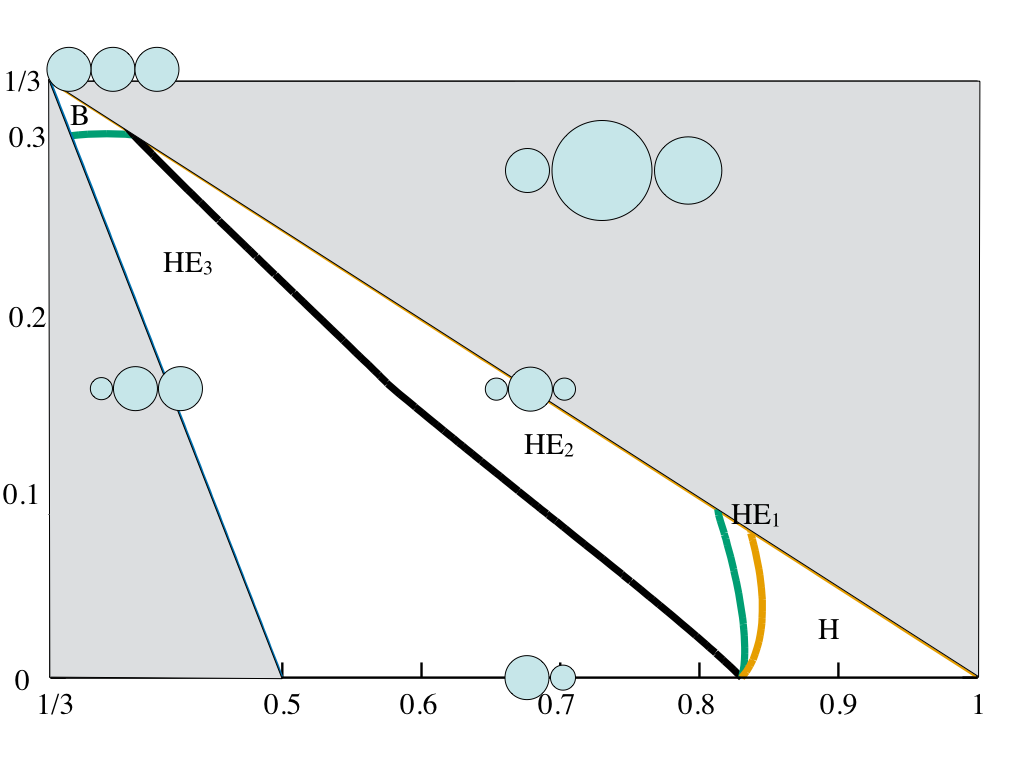}
\caption{Diagram showing mass fractions for different Hill Stability conditions to be satisfied for the ER312 configuration fission condition. The highest mass fraction that can be ejected is seen to be 0.304.}
\label{fig:ERH312}
\end{figure}

\clearpage

\subsection{Energy -- Angular Momentum Diagrams}

The mass triangle diagrams are able to clearly convey the conditions under which a fissioned system can undergo escape, and in indicating when the ER conditions exist. However, to convey when the EA configurations exist is more complex, and better shown through plots showing the relative energy of the different stable configurations as a function of angular momentum. 
These energy -- angular momentum plots were defined in earlier publications (see e.g., \cite{scheeres_minE}). Figure \ref{fig:EH_study} shows a detailed and deconstructed energy -- angular momentum plot for the case of $m_1=1/2$, $m_2=1/3$ and $m_3=1/6$, to indicate how the LR, ER and EA configurations are related to each other and to the necessary conditions. The combined plot, in the upper left, shows all of the different stable energy states and creates a complex picture as many of the states are seen to exist at similar levels of angular momentum and energy. Note that even though existence curves may cross, the different relative equilibrium configurations are geometrically distinct from each other. The ER configurations are stationary when they exist, whereas the distance between the EA configurations will vary as the angular momentum is increased. We note that the EA configuration curves all have two branches after their bifurcation value in $H^2$. The lower branch is the stable relative equilibrium with an increasing separation as the angular momentum increases. The upper branch is an unstable branch with a relative distance that decreases with an increase in angular momentum and either ends by intersecting with the ER configuration, destabilizing it, or that intersects with the unstable Euler Orbital relative equilibrium (not shown). We can note a few additional points. First, we see that the energy -- angular momentum curves of the EAij-k and EAji-k configurations bifurcate at different conditions but become indistinguishable for larger angular momenta. Also, it is interesting to note that the EAij-k curves asymptotically approach the $\mathbf{HE}_k$ necessary conditions for escape as $H^2 \rightarrow \infty$. 

\begin{figure}[htb]
\centering
\includegraphics[scale=0.35]{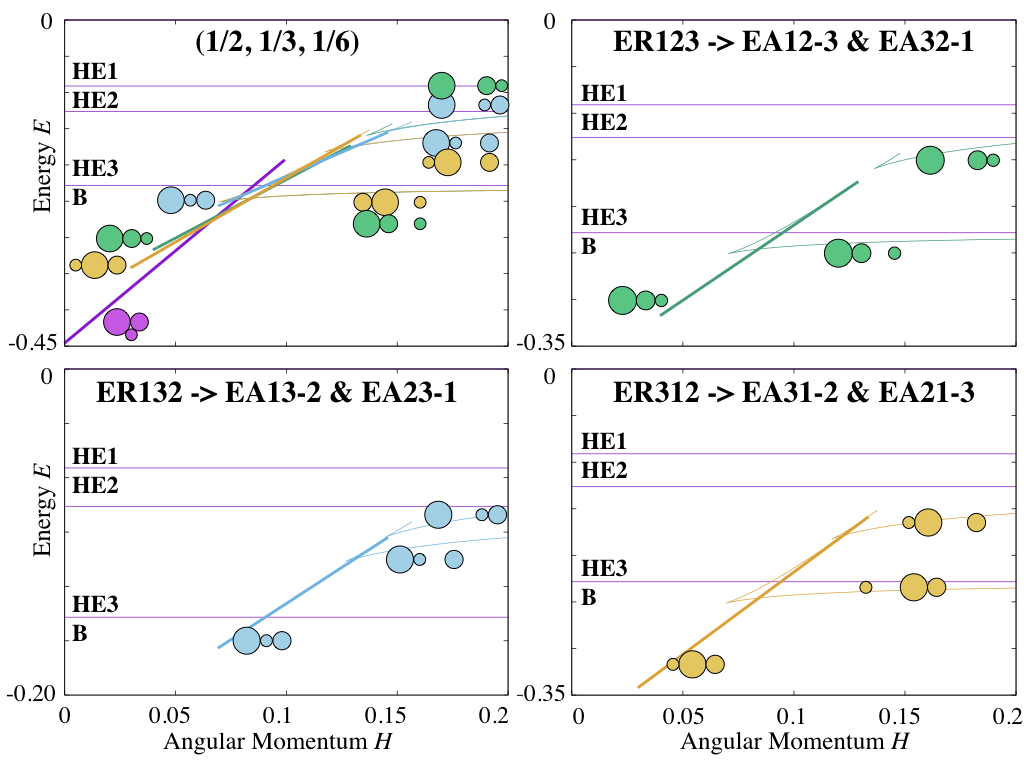}
\caption{Composite and deconstructed Energy-Angular momentum graph of a system with mass fractions 1/2, 1/3, 1/6.}
\label{fig:EH_study}
\end{figure}

In Fig.\ \ref{fig:survey} we show a survey of energy -- angular momentum charts for a number of different mass fractions chosen across the parameter space. These are to show the diversity of relations that exist between the limits on the ER configurations, the necessary conditions for escape, and when the EA configurations may exist relative to the ER termination points. We note that along the borders of the parameter space some of the distinct cases become the same as two of the bodies will then have equal values of mass, cutting down on the number of distinct curves. Points of interest include mapping the transition of the LR and ER termination points relative to each other and relative to the necessary conditions for escape. 

\begin{figure}[htb]
\centering
\includegraphics[scale=0.35]{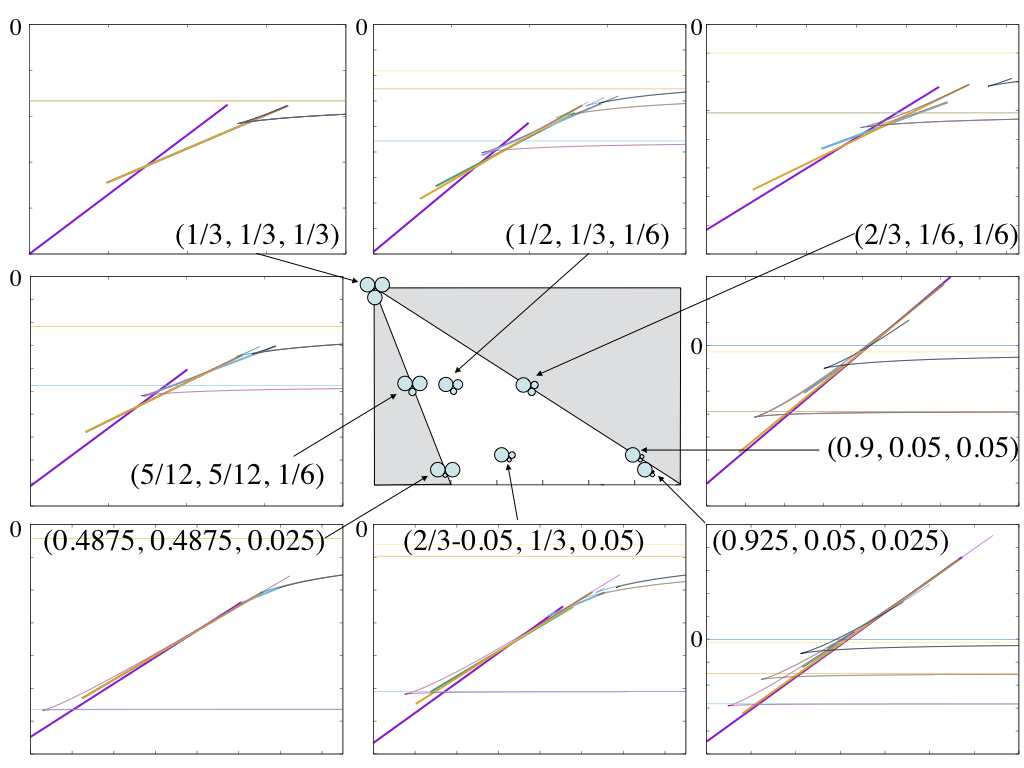}
\caption{Collection of Energy-Angular Momentum plots showing the range of behaviors across the parameter space.\textcolor{black}{On each sub-graph the zero energy level is indicated, to show the shifting energy levels across the region.}}
\label{fig:survey}
\end{figure}

\subsection{Post-Escape System State}

If the original 3-body system satisfies the necessary escape conditions, and if one of the bodies is ejected, it is possible to put constraints on what state the remaining 2-body system can lie in. Here we nominally assume that no energy has been dissipated, and instead focus on the range of possible outcomes for the disrupted system. 

As a general model, we assume that the energy of the disrupted system is described by the mutual escape kinetic energy between the bound pair and the unbound body plus the energy of the bound pair relative to each other. We assume that the bound pair lies in a minimum energy configuration for the full 2-body problem, which can either be a resting configuration or a circular, doubly synchronous orbit. This state can be generically defined as
\beq
	E_f & = & \frac{1}{2} m_k (m_i + m_j) V_{\infty}^2 + \frac{1}{2} \left[ I_{S_i} + I_{S_j} + \frac{m_i m_j}{m_i+m_j} d_{ij}^2 \right] \Omega^2 - \frac{m_i m_j}{d_{ij}}  \label{eq:Ef} \\ 
	& = & \frac{1}{2} m_k (m_i + m_j) V_{\infty}^2 + E_2
\eeq
where the kinetic energy term is divided by the total mass, $m_i + m_j + m_k = 1$, $E_2$ denotes the energy of the remaining full 2-body system, the distance $d_{ij} \ge 1-r_k$, and the term $I_{S_i}$ corresponds to the moment of inertia of body $i$. 

We recall that for the full 2-body problem, the system will fission if their mutual spin rate equals $\Omega^2 = \frac{m_i + m_j}{(r_i+r_j)^3}$. This corresponds to an energy in the full 2-body system of
\beq
	E_{2f} & = & \frac{1}{2} \frac{m_i + m_j}{(r_i+r_j)^3} \left[ I_{S_i} + I_{S_j} - \frac{m_i m_j}{m_i+m_j} (r_i+r_j)^2 \right] \label{eq:E2f}
\eeq
Also, mapping the calculation from \cite{scheeres_minE} into the current units, if the energy of the full 2-body problem exceeds the value 
\beq
	E_{2b} & = & - \frac{ (m_i m_j)^{3/2} }{3 \sqrt{3} \sqrt{ (m_i+m_j) (I_{S_i} + I_{S_j})}} \label{eq:E2b}
\eeq
an orbital relative equilibrium exists. We note that $E_{2b} < E_{2f}$. At energies less than $E_{2b}$, only a resting equilibrium orbit exists. For energies in the range $E_{2b} \le E_2 \le E_{2f}$ both a resting and an orbital configuration exist. And at energies greater than $E_{2f}$, only an orbital configuration exists. 

To find an upper limit on the escape speed, assume that $\Omega = 0$, leaving all of the bodies non-rotating and bodies $i$ and $j$ resting on each other. Solving for the escape speed gives the maximum possible
\beq
	V_M & = & \sqrt{ \frac{2}{m_k(m_i+m_j)} \left[ E_f + \frac{m_i m_j}{1-r_k} \right] }
\eeq
so that $V_{\infty} \le V_M$. 
It is clear that if the necessary condition $\mathbf{HE}_k$ is satisfied, then the term inside the square root will be positive and well defined.  

At the other extreme, assume that $V_{\infty} = 0$ and that the remainder of the energy is deposited in the bound body's orbit and rotation. In this case we see that the fission energy of the 3-body system equals the energy of the bound 2-body system, or $E_f = E_2$. To place constraints on the bound 2-body system, we can compare the fission energies to the limiting 2-body energies given in Eqns.\ \ref{eq:E2f} and \ref{eq:E2b}. Note that for each fission energy we compare the energy $E_f$ with three different values for $E_2$, corresponding to bodies 1 and 2, 1 and 3, or 2 and 3 being bound. 

Figure \ref{fig:CB} presents these comparisons for the fission energy of the LR and the three ER configurations. There is a similar pattern for all cases. First, for all configurations we find that $E_f \ge E_{2b}$ for cases with small enough masses for bodies 2 and 3. However, we only find that $E_f \ge E_{2f}$ for the case when body 1 is ejected and bodies 2 and 3 remain. For all other cases, there are both resting and orbital states possible for the bound 2-body system. For the region where $E_f < E_{2b}$ only resting configurations are possible after ejection of body $k$. We note that we have plotted the binary conditions independent of whether the corresponding $\mathbf{HE}_k$ necessary condition is satisfied. That being said, there is some correspondence between the necessary conditions and the different binary outcomes, the most clear one being that when the $\mathbf{H}$ necessary condition is satisfied the BO23 condition (see the caption of Fig.\ \ref{fig:CB} for a description of this notation) is also satisfied, meaning that if the smaller two bodies are left in a bound condition they must orbit each other. 

\begin{figure}[htb]
\centering
\includegraphics[scale=0.35]{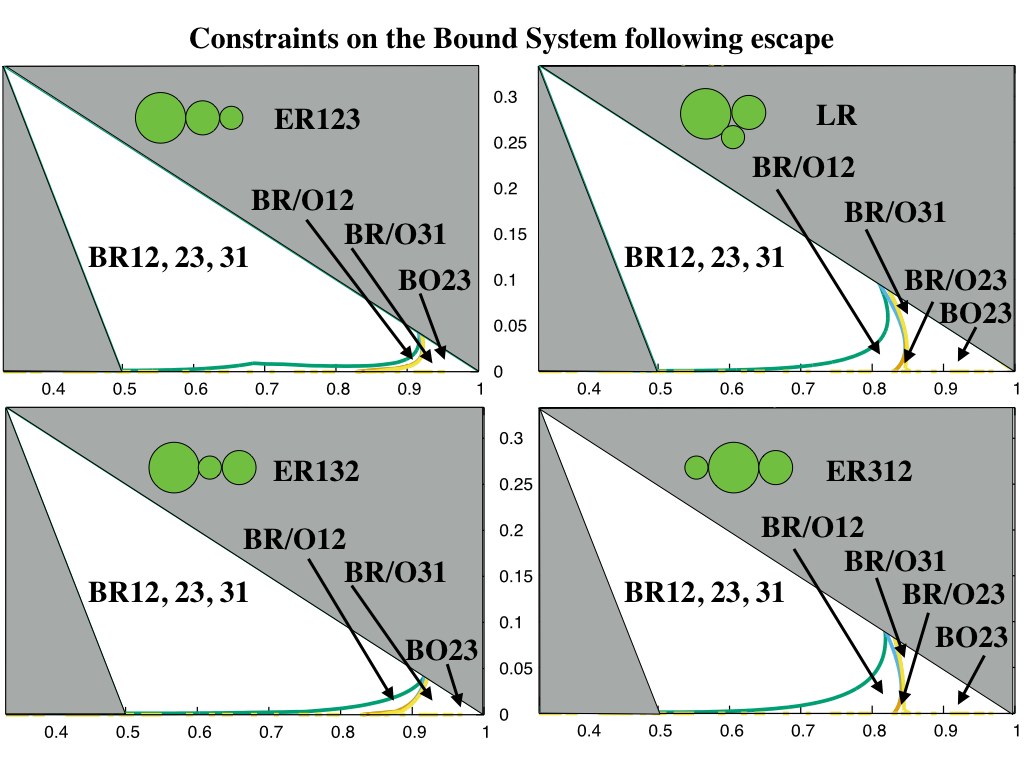}
\caption{Constraints on the bound system following escape of body $k$. BRij denotes that bodies $i$ and $j$ can only exist in a resting configuration after body $k$ has escaped, BR/Oij denotes that the bodies can either be in a resting or orbiting configuration after body $k$ has escaped, and BOij denotes that the bodies can only be in an orbiting configuration. Note that the two panels on the right hand side (LR and ER312) have a small visible region where the BR/O23 case occurs. In all of the other panels the region where both outcomes occur are too small to be seen, even though they are present. }
\label{fig:CB}
\end{figure}

\subsection{Sufficiency Condition for $\mathbf{HE}$ or $\mathbf{H}$}

There are only two sufficiency conditions identified in Theorem \ref{thm:1}, for bounded motion and for guaranteed escape. While the former occurs, the latter will never apply to systems that undergo fission, meaning that it is never guaranteed that a fissioned system that satisfies some of the necessary conditions for escape will also satisfy the sufficient condition for escape. To establish this, we note that the proof rests on the quantity $E - T_r = T_o + \mathcal{U} > 0$. This can be explicitly checked using the LR and ER fission spin rates from Eqns.\ \ref{eq:LRFSR} and \ref{eq:ERFSR} for the corresponding configurations. Computing this quantity across the parameter space explicitly shows that this quantity is always positive. Thus, when any LR or ER configuration fissions it is never guaranteed that one of the bodies will escape, this is so even though it is possible for one of these bodies to escape. 

\section{Summary and Discussion}

In this paper we explore the energetics of the spherical, full 3-body problem with a specific focus on whether bodies that undergo fission are able to dynamically eject components. These energetics were found to be important in developing a basic understanding of the formation conditions for asteroid pairs, and this study shows that accounting for additional bodies in the system can eject a significantly larger fraction of the initial mass of the system, with the mass fraction increasing up to 0.31 as compared to a 0.17 limit in the binary case. However, the parameter space that leads to such a larger ejected mass is somewhat limited. This limit can be explicitly compared with figure 6 from \cite{AIV_interiors} (note that a mass fraction of 0.31 corresponds to a mass ratio between escaping mass and remaining mass of 0.45), and thus such 3rd body effects could account for those asteroid pairs that violate the 2-body limits previously identified. 

It is important to note that the current paper only focuses on the energetics of fission and escape, and does not investigate the detailed dynamics of these fissioned systems. Such studies are necessary, and in general will need to include the effects of non-spherical shapes and surface forces whenever two bodies collide. These studies are necessary to evaluate the sharpness of the conditions found herein, similar to what was previously done for the full 2-body problem in \cite{scheeres_F2BP_planar, jacobson_icarus}. 

\section*{Acknowledgements}

The author acknowledges support from NASA grant NNX14AL16G from the Near Earth Objects Observation programs. 

\bibliographystyle{plain}
\bibliography{../../../bibliographies/biblio_article,../../../bibliographies/biblio_books,../../../bibliographies/biblio_misc}

\end{document}